\documentclass[draftcls,12pt,onecolumn]{IEEEtran}

\usepackage{times,amsmath,color,amssymb,graphicx,epsfig,cite,psfrag,enumerate}
\usepackage{subfigure,multirow}

\newtheorem{Remark}{\it Remark}[section]

\newtheorem{Proposition}{\it Proposition}[section]

\begin{document}
%
\title{Asymptotic Capacity of Large Fading Relay Networks
with Random Node Failures}


\author{\IEEEauthorblockN{Chuan Huang,~\IEEEmembership{Student Member,~IEEE,}~Jinhua Jiang,~\IEEEmembership{Member,~IEEE,} \\ Shuguang
Cui},~\IEEEmembership{Member,~IEEE,}

\thanks{Part of the work was presented at the Forty-Seventh Annual
Allerton Conference on Communication, Control, and Computing, and
GLOBECOM'09.}

\thanks{Chuan Huang and Shuguang Cui are with the Department of Electrical
and Computer Engineering, Texas A\&M University, College Station,
TX, 77843. Emails: \{huangch, cui\}@tamu.edu.}

\thanks{Jinhua Jiang is with the Department of Electrical Engineering,
Stanford University, Stanford, CA, 94305, Email:
jhjiang@stanford.edu.}}

\maketitle
\begin{abstract}
To understand the network response to large-scale physical
attacks, we investigate the asymptotic capacity of a half-duplex
fading relay network with random node failures when the
number of relays $N$ is infinitely large. In this paper, a
simplified independent attack model is assumed where each relay
node fails with a certain probability. The noncoherent relaying scheme is considered, which corresponds to the case of zero forward-link channel state information (CSI) at the
relays. Accordingly, the whole
relay network can be shown equivalent to a Rayleigh fading
channel, where we derive the $\epsilon$-outage capacity upper
bound according to the multiple access (MAC) cut-set, and the
$\epsilon$-outage achievable rates for both the
amplify-and-forward (AF) and decode-and-forward (DF) strategies. Furthermore, we show that the DF strategy is
asymptotically optimal as the outage probability $\epsilon$ goes
to zero, with the AF strategy strictly suboptimal over all signal to noise ratio (SNR)
regimes. Regarding the rate loss due to random attacks, the AF strategy suffers a less portion of rate loss than
the DF strategy in the high SNR regime, while
the DF strategy demonstrates more robust performance in the low SNR regime.
\end{abstract}




%
\IEEEpeerreviewmaketitle
\section{Introduction}

Node cooperation has been shown as an effective way to improve
system capacity and provide diversity in wireless networks. One of
the promising cooperation schemes is the use of relays, while the
capacity of general relay channels is still an open problem. The
full-duplex relay channels, in which relay nodes can transmit and
receive simultaneously, have been intensively investigated, e.g.,
in \cite{cover}, \cite{xie}, and \cite{kramer}, where various
achievable rates and special-case capacity results have been
obtained. The more practical case based on half-duplex relays has
also been intensively studied, which will be reviewed next.

\subsection{Related Works}

In practice, relay nodes can only work in a half-duplex mode,
which means that they cannot transmit and receive simultaneously
over the same frequency band. The half-duplex relay channels using
the decode-and-forward (DF) and compress-and-forward (CF)
strategies have been studied in \cite{gamal} and \cite{madsen} and
the references therein. In \cite{tse_high}, the authors
investigated the outage probability of fading half-duplex relay
channels using both the DF and amplify-and-forward (AF) strategies in
the high signal to noise ratio (SNR) regime; in \cite{tse_low}, the asymptotic
performance of the DF and AF strategies in the low SNR regime was studied,
and the burst AF (BAF) strategy was proved optimal in special
scenarios.

The parallel relay channel, which only contains two relay nodes, was first studied in \cite{schein},
where the achievable rates are obtained by using the AF and DF
strategies, and power sharing between these two strategies
was shown able to achieve a higher rate. In \cite{xue}, the authors discussed two
time-sharing schemes and various relay strategies, e.g., AF, DF,
and CF.

In \cite{Gastpar1,Gastpar2,Gastpar3,blocskei,coso}, the asymptotic
behavior of half-duplex large relay networks was studied.
Considering the joint source channel coding problem for a class of
Gaussian relay networks \cite{Gastpar1,Gastpar2,Gastpar3}, the
capacity is shown to be asymptotically achievable with the AF strategy
when the number of relays tends to infinity. However, for fading
relay networks, the capacity is still unknown. In \cite{blocskei},
the authors extended these results to fading multiple-input and
multiple-output (MIMO) relay networks, where both the coherent and
noncoherent relaying schemes are discussed, assuming perfect
channel state information (CSI) and zero CSI at the relays,
respectively. For the coherent case with CSI, the authors showed
that the achievable rate scales as $\mathcal{O}(\log(N))$ in the
high SNR regime; for the noncoherent case without CSI, the relay
networks were shown to behave as a point-to-point MIMO channel in
the high SNR regime. In \cite{coso}, the authors studied the
scaling laws of the AF, DF, and CF strategies for Gaussian
parallel relay networks.

\subsection{Our Contributions}

In this paper, we study half-duplex fading relay networks
and impose a total power constraint across all nodes, i.e., the
total transmit power consumed in the whole network is less than a
finite value $P$. Furthermore, both the AF and DF strategies are
investigated; and we assume that there is zero forward-link
(relay-to-destination link) CSI at the relays, which leads to the
noncoherent relaying scheme. We focus
on an unreliable networking scenario that takes into account the
random failures of the relays, where each relay is prone to random
physical attacks such as wild fire or power losses \cite{chuan,chuan2}. In such
situations, we may know the total number of relays, but may not
know which subset of them is actually functioning. As such, we
adopt a simplified model where each of the relays is independently
failing with a probability $p$ such that the number of operable
relays is an unknown random variable. Under such a setup, we study
the asymptotic capacity upper bound and achievable rates for this
relay system when the number of relays $N$ is infinitely large: For the
noncoherent scheme, this network can be shown equivalent to a
Rayleigh fading channel, and we study the $\epsilon$-outage rates
achieved by various relay strategies. We summarize
the main results of the paper as follows:
\begin{enumerate}
    \item For the case with zero forward-link CSI at the relays, we derive the $\epsilon$-outage capacity upper bound via the multiple access (MAC)
cut-set, the AF $\epsilon$-outage rate, and the DF
$\epsilon$-outage rate, all of which scale on the order of
$\mathcal{O}(\log(\gamma_0))$ as $\gamma_0$ gets large, with $\gamma_0$ the transmit SNR.

\item For the AF strategy, we quantify
the gap between the achievable rate and the MAC bound in the high
SNR regime; in the low SNR regime, we show that the AF strategy
performs poorly. Moreover, for general $\gamma_0$ values, we show that it
is a quasiconcave problem to determine the optimal power
allocation between the source and the relays. For the DF strategy,
we prove that it is asymptotically optimal as the outage
probability $\epsilon$ goes to zero. Moreover, we derive the closed-form expression for the
optimal power allocation factor when $\epsilon$ is small.
    \item Regarding the effect of random node failures, we derive the upper
    bound of the achievable rate loss.
    Moreover, it is proved that the AF strategy is less sensitive to random attacks than the DF strategy in the high SNR
    regime; the reverse is true
    for the low SNR regime.
\end{enumerate}

It is worth pointing out that in \cite{chuan2}, we investigated the associated scaling law and power allocation problem for the DF strategy in the coherent relay case, which assumes perfect forward-link CSI at the relays.

\subsection{Organization and Notations}

The remainder of the paper is organized as follows. In Section II,
we present the network, channel, and signal models, and also
introduce all the assumptions. In Section III, we derive the capacity upper bound using MAC cut-set. In Section IV, we focus on the achievable rates for both the AF and DF strategies, and
discuss their performances in both the high and low SNR regimes. In section V, we present some simulation and numerical results. Finally, the paper
is concluded in Section VI.

Notation: we define the following notations used throughout this
paper.
\begin{itemize}
    \item $\text{dist}(\mathbf{a},\mathcal{S})$ is the distance
    from a point $\mathbf{a}$ to a set $\mathcal{S}$, which is
    defined as $\min_{\mathbf{b} \in \mathcal{S}} \| \mathbf{a} -\mathbf{b}
    \|$ with $\| \mathbf{a} -\mathbf{b}
    \|$ a distance measure between two points $\mathbf{a}$ and
    $\mathbf{b}$;
    \item $f(x)\sim g(x)$ means $\lim_{x \rightarrow
    +\infty}\frac{f(x)}{g(x)}=1$;
    \item $X \sim \mathcal{O}(N)$ means $\lim_{N \rightarrow +\infty}
    \frac{X}{N}=C$, where $C$ is a finite positive constant;
    \item $A \sim \mathcal{O}(1)$ means $A$ is a bounded constant;
    \item $\log(x)$ and $\ln(x)$ are the base-2 and natural logarithms, respectively;
    \item $\mathbb{E}(X)$ is the expectation of a random variable $X$.
\end{itemize}

\section{Assumptions and System Model} \label{system_model}

\subsection{General Assumptions}
We consider a relay network with a pair of source and destination
nodes, which are assumed to be located at two fixed positions
$\mathbf{s}_S$ and $\mathbf{s}_D$ in a given region, respectively,
and $N$ relay nodes independently and identically distributed
(i.i.d.) over a given area $\mathcal{S}$ with a probability
density function (PDF) $p(\mathbf{s}),~\mathbf{s} \in
\mathcal{S}$. Under a dead-zone assumption, we let
$\text{dist}(\mathbf{s}_S,\mathcal{S}) \geq s_0$ and
$\text{dist}(\mathbf{s}_D,\mathcal{S}) \geq s_0$, where $s_0>0$ is
a positive constant to define the radius of a dead-zone. In the dead-zone, communication is not permitted to
ensure that the received power is always bounded. Under the above
setting, the 2-D network topology is shown in Fig. \ref{fig1}.

We assume that each node has only one antenna, the relays work in
the half-duplex AF or DF mode, and the transmissions follow a
time-slotted structure. Specifically, each time slot is split into
two parts: In the first part, the source broadcasts the message to
all the relays; in the second part, the relays transmit certain
messages to the destination based on what they received in the
first half. Generally, for the DF strategy, we could allocate
different time fractions to these two parts to optimize the system
performance \cite{madsen}, while keep an equal time allocation for
the AF strategy. For simplicity, even for DF we assume that these
two parts are of equal length. In addition,
we assume no direct transmissions between the source and the
destination due to their relatively large distance.

For the AF strategy, the relays simply forward a scaled version of the
received signal. For the DF strategy, the relays decode the source
message based on the backward-link (source-to-relay link) CSI: If
a relay cannot decode the source message successfully, it keeps
silent in the forwarding phase; otherwise, it re-encodes the
decoded message with the same codebook as used in the source node
and transmits the codeword to the destination. Due to the fading effects in the
backward-links, it is possible that the relays cannot successfully
decode the source message, which is regarded as a type of
fading-related decoding failures. Moreover, with random physical
attacks, we assume that each relay may die with a probability $p$
independently within the time interests of operation. Accordingly,
we further assume that these two types of random failures are
statistically independent. Thus, the overall number of successful
functioning relays (denoted as $L$) can be modelled as a binomial
random variable with parameter $\left( N,p_0(1-p) \right)$, where
$p_0$ is the probability of successfully decoding in the relays,
with $0 \leq p_0 \leq 1$.

We also impose a constraint on the total transmit power consumed
by the source and all the relays. In particular, we set the total
transmit power spent in the network as a finite value $P$, denote
the power allocation fraction to the source node as $\alpha$ with
$\alpha \in [0,1)$, and allocate the rest $(1-\alpha)P$ power to
all the relays.

\subsection{Channel and Signal Models}
In the first half of a given time slot, the channel input-output
relationship between the source and the $i$-th relay (located at
position $\mathbf{s}_i$) is given as
\begin{eqnarray} \label{t1}
r_i=\sqrt{\alpha P\rho_{Si}}h_{Si}x+n_i,~i=1,2,\cdots,N,
\end{eqnarray}
where $r_i$ is the received signal at the $i$-th relay, $x$ is the
unit-power symbol transmitted by the source, $\alpha P$ is the
transmit power allocated to the source, the short-term fading
coefficients $h_{Si}$'s are i.i.d. zero mean complex Gaussian
random variables with unit variance, i.e., $h_{Si}\sim
\mathcal{CN}(0,1)$, $\rho_{Si}$ is the average channel power gain
of the $i$-th source-to-relay link due to the long-term path-loss
with $\rho_{Si}=\frac{1}{\|\mathbf{s}_S - \mathbf{s}_i
\|^{\theta}}$, where $\theta$ is the path loss exponent \cite{tse}
and $\|\mathbf{x}\|$ is the Euclidean norm of vector $\mathbf{x}$,
and $n_i$ is the i.i.d. complex additive white Gaussian noise
(AWGN) with distribution $\mathcal{CN}(0,N_0)$.

In the second half of a given time slot, a symbol $t_i$ is
transmitted from the $i$-th relay to the destination. Based on
different forward-link CSI assumptions and various relay
strategies, we can properly design $t_i$'s to achieve the best
system performance. Since we assume that there is no direct link
between the source and the destination, the received signal $y$ at
the destination is given as
\begin{eqnarray} \label{t2}
y=\sum_{i=1}^{L} {\sqrt{\rho_{iD}} h_{iD} t_i} + w,
\end{eqnarray}
where $L$ is a binomial random variable with $\{1,2,\cdots,L\}$
denotes the subset of successful functioning relays, $h_{iD}$'s
are i.i.d. with $h_{iD}\sim \mathcal{CN} (0,1)$, $\rho_{iD}$ is
the path-loss of the $i$-th forward-link with $\rho_{iD}=\frac{1}
{\| \mathbf{s}_i - \mathbf{s}_D\|^{\theta}} $, and $w$ is the
complex AWGN with distribution $\mathcal{CN} (0,N_0)$. Note that
here message $t_i$ may not have a unit instantaneous power and the
long term average sum transmit power across all the relays is
equal to $(1-\alpha)P$.

\section{MAC Cut-Set Upper Bound}

We now consider a large relay network where only the backward-link
CSI is available at the relays, i.e., no knowledge about the
forward-link CSI. In this case, coherent transmission is
impossible for the second-half time slot. Therefore, each relay
should transmit with identical power level to achieve the optimal
performance, i.e., the power allocated to each relay is
$\frac{(1-\alpha)P}{N}$. Accordingly, we adopt the
$\epsilon$-outage rate (see Chapter 5 of \cite{tse}) as the performance criterion,
which is defined over a target transmission outage probability
$\epsilon$.

In this section, we derive an upper bound on the $\epsilon$-outage
capacity \cite{tse} of the large relay network defined in the
previous section. For a given finite $N$, it is difficult to
obtain an explicit expression for the outage rate upper bound. As
such, we focus on the case where $N$ is infinitely large, and
investigate its asymptotic behavior.

\begin{Proposition} \label{prop_non_upp}
Without the forward-link CSI at the relays, when $N$ goes to infinity,
the $\epsilon$-outage capacity upper bound of the large relay
network defined in Section \ref{system_model} is asymptotically
given by
\begin{eqnarray} \label{capacityupper}
C_{\text{upper}} = \frac{1}{2} \log \left( 1+
\gamma_{\text{upper}} \ln\left( \frac{1}{1-\epsilon} \right)
\right),
\end{eqnarray}
where $\gamma_{\text{upper}}$ is the upper bound of the average
received SNR at the destination given as
\begin{eqnarray} \label{snrupper}
\gamma_{\text{upper}}=(1-p)(1-\alpha)\gamma_0\mathbb{E}(\rho_{iD}),
\end{eqnarray}
with $\gamma_0=\frac{P}{N_0}$ and $\mathbb{E}(\rho_{iD})=\int_{\mathbf{s}\in \mathcal{S}}
\rho_{iD}(\mathbf{s}) p (\mathbf{s})d\mathbf{s}$.
\end{Proposition}

\begin{proof}
See Appendix \ref{app_non_upp}.
\end{proof}

\begin{Remark}
In general, the upper bound in (\ref{capacityupper}) is loose, but
can be achieved when the outage
probability target $\epsilon$ tends to zero. We will show that this upper bound is
asymptotically tight in Section \ref{non_df}.
\end{Remark}

Note that in this paper whenever we refer to the low or high SNR regimes, we refer to the low or high values of $\gamma_0$ defined above.

\section{Achievable Rates for Large Relay Networks}

\subsection{AF Strategy}

In the previous section, we derived an upper bound of the
$\epsilon$-outage capacity for the large fading relay network. In
this subsection, we derive the $\epsilon$-outage rate with the AF
strategy, discuss the optimal power allocation between the source
and the relays, and evaluate the performance of the AF strategy in
both the low and high SNR regimes. Moreover, the rate loss due to
random attacks is evaluated.

\subsubsection{Achievable Rate}

We assume that the $i$-th relay node could estimate the average
power of the received signal: $\mathbb{E}\left[|r_i|^2 \big|
\rho_{Si} \right] =\alpha P \rho_{Si} + N_0$, and performs the
amplification according to
\begin{eqnarray} \label{t3}
t_i=\frac{\sqrt{(1-\alpha) P}}{ \sqrt{N(\alpha P\rho_{Si}+N_0)}}
r_i,
\end{eqnarray}
which ensures the sum power constraint satisfied across the relays:
$\mathbb{E}\left[\sum_{i=1}^{N}|t_i|^2\right]= (1- \alpha ) P$.
Hence, from (\ref{t1}), (\ref{t2}), and (\ref{t3}), the received
signal at the destination is given as
\begin{align} \label{receivesig}
y=&\underbrace{ \left( \sum_{i=1}^{L}
\sqrt{\frac{\alpha(1-\alpha)P^2\rho_{Si} \rho_{iD} }{N(\alpha P
\rho_{Si}+N_0)}} h_{Si}h_{iD}\right) }_{A} x +\underbrace{ \left(
\sum_{i=1}^{L} \sqrt{\frac{(1-\alpha)P \rho_{iD} }{N(\alpha P
\rho_{Si}+N_0)}}h_{iD} n_i\right) }_{B} +w.
\end{align}
From \cite{robbins}, we know that when $N \rightarrow \infty$, the
distributions of $A$ and $B$ are asymptotically given by
\begin{align}
&A \sim  \mathcal{CN} \left( 0, (1-p) \mathbb{E}\left[
\frac{\alpha(1-\alpha)P^2\rho_{Si} \rho_{iD} }{\alpha P
\rho_{Si}+N_0} \right] \right), \label{adis} \\
&B  \sim \mathcal{CN} \left( 0, (1-p) \mathbb{E}
\left[\frac{(1-\alpha)P \rho_{iD} }{\alpha P \rho_{Si}+N_0} N_0
\right] \right). \label{bdis}
\end{align}

\begin{Remark}
Since $h_{Si}$ and $h_{iD}$ are independent and of zero means, we
have $\mathbb{E}(AB)=\mathbb{E}(A)\mathbb{E}(B)=0$, which means
that $A$ is uncorrelated with $B$. Since $A$ and $B$ are
asymptotically complex Gaussian, they are independent of each
other. Therefore, from (\ref{receivesig}), the large fading relay
network under consideration is asymptotically equivalent to a
Rayleigh fading channel between the source and the destination
with a fading coefficient $A$ and an AWGN $B+w$.
\end{Remark}

Thus, the average received SNR at the destination can be written
as
\begin{eqnarray} \label{snraf}
\gamma_{\text{AF}}=\frac{ (1-p) \mathcal{A} } { 1+ (1- p)
\mathcal{B} },
\end{eqnarray}
where $\mathcal{A}$ and $\mathcal{B}$ are defined as
\begin{eqnarray}
\mathcal{A}=\mathbb{E}\left(
\frac{\gamma_{Si}\gamma_{iD}}{1+\gamma_{Si}} \right)
\label{A},~\mathcal{B} =\mathbb{E} \left ( \frac{\gamma_{iD}
}{1+\gamma_{Si}}\right ), \label{B}
\end{eqnarray}
with $\gamma_{Si}$ and $\gamma_{iD}$ the received SNRs of the
$i$-th source-to-relay and relay-to-destination link,
respectively: $\gamma_{Si}=\alpha \gamma_0 \rho_{Si}$; and
$\gamma_{iD}=(1-\alpha) \gamma_0 \rho_{iD}$. Therefore, with the
same argument as for the case of the $\epsilon$-outage upper bound, the $\epsilon$-outage rate with
the AF strategy is given by
\begin{eqnarray}\label{capacityaf}
R_{\text{AF}}=\frac{1}{2} \log \left( 1+\gamma_{\text{AF}}
\ln\left( \frac{1}{1-\epsilon} \right)\right),
\end{eqnarray}
where $\epsilon$ is the target outage probability.

\subsubsection{Effect of Random Attacks}
From (\ref{snraf}), (\ref{A}), and (\ref{capacityaf}), it is
observed that $R_{\text{AF}}$ is a decreasing function of $p$,
which means that the random attacks cause a certain amount of rate
loss compared with the attack-free case, i.e., when $p=0$. Based
on (\ref{capacityaf}), we next present some analysis on the attack
effects over both the low and high SNR regimes.

\begin{Proposition}
In the low SNR regime, there is a $p$-portion rate loss for the AF
strategy.
\end{Proposition}
\begin{proof}
With low SNR values, we have $1+\gamma_{Si} \approx 1$ and
$1+\mathcal{B}\approx 1$; thus,
\begin{align}
R^{\text{Low}} \approx \frac{1}{2}(1-p) \mathbb{E} \left(
\gamma_{Si} \gamma_{iD} \right) \ln\left( \frac{1}{1-\epsilon}
\right) = (1-p)R_{\text{0}}^{\text{Low}},
\end{align}
where $R_{\text{0}}^{\text{Low}}=  \frac{1}{2} \mathbb{E} \left(
\gamma_{Si} \gamma_{iD} \right) \ln\left( \frac{1}{1-\epsilon}
\right) $ is the $\epsilon$-outage rate with the AF strategy at
$p=0$ in the low SNR regime.
\end{proof}

\begin{Remark}
In this case, random attacks cause a $p$-portion rate loss
compared to the attack-free case. Intuitively, since the random
attacks kill $p$-portion of the relays, the received signal power
at the destination is also $p$-portion of that in the attack-free
case. Accordingly, the rate is lost at the same proportion since we are in the
low SNR regime.
\end{Remark}

\begin{Proposition} \label{losshigh}
In the high SNR regime, there is a constant rate loss for the AF
strategy, and this constant is upper-bounded by $\frac{1}{2}\log
\left(\frac{1}{1-p}\right)$.
\end{Proposition}
\begin{proof}
Since $\mathcal{A}\sim \mathbb{E}\left( \gamma_{iD}\right)$ and
$\mathcal{B} \sim \mathbb{E}
\left(\frac{\gamma_{iD}}{\gamma_{Si}}\right)$ in the high SNR
regime, we have
\begin{align}
R^{\text{High}} &\sim \frac{1}{2}\log \left(
\frac{(1-p)(1-\alpha)\gamma_0\mathbb{E}\left( \rho_{iD}\right)
}{1+(1-p) \frac{1-\alpha}{\alpha} \mathbb{E}\left(
\frac{\rho_{iD}}{\rho_{Si}} \right) } \ln\left(
\frac{1}{1-\epsilon} \right) \right) \nonumber \\
&= R_{\text{0}}^{\text{High}} -\frac{1}{2}\log
\left(\frac{1}{1-p}\right)+ \underbrace{\frac{1}{2}\log \left(
\frac{1+\frac{1-\alpha}{\alpha}\mathbb{E}\left(
\frac{\rho_{iD}}{\rho_{Si}}\right)}{1+(1-p)\frac{1-\alpha}{\alpha}
\mathbb{E}\left( \frac{\rho_{iD}}{\rho_{Si}}\right)} \right) }_{C_1} \label{non_af_high_1} \\
& \geq R_{\text{0}}^{\text{High}} -\frac{1}{2}\log
\left(\frac{1}{1-p}\right), \label{non_af_high_2}
\end{align}
where $R_{\text{0}}^{\text{High}}= \frac{1}{2}\log \left(
\frac{(1-\alpha)\gamma_0\mathbb{E}\left( \rho_{iD}\right) }{1+
\frac{1-\alpha}{\alpha} \mathbb{E}\left(
\frac{\rho_{iD}}{\rho_{Si}} \right) } \ln\left(
\frac{1}{1-\epsilon} \right) \right) $ is the $ \epsilon$-outage
rate with the AF strategy at $p=0$ in the high SNR regime. From
(\ref{non_af_high_2}), it is observed that the upper bound of the
rate loss is given as $R_{\text{0}}^{\text{High}} -
R^{\text{High}} = \frac{1}{2}\log \left(\frac{1}{1-p}\right)-C_1
\leq \frac{1}{2}\log \left(\frac{1}{1-p}\right)$, and the equality
holds only when $C_1=0$, i.e., $p=0$.
\end{proof}

\subsubsection{Power Allocation to the Source}

In this subsection, we mainly discuss the power allocation
strategy to maximize the achievable rate with the AF strategy. In
order to maximize (\ref{capacityaf}), we only need to maximize
(\ref{snraf}). As such, we have the following results.
\begin{Proposition} \label{opt_alpha_non_af}
The received SNR defined in (\ref{snraf}) is a quasiconcave
function over $\alpha$.
\end{Proposition}
\begin{proof}
See Appendix \ref{app_non_af_power}.
\end{proof}

Since (\ref{snraf}) is quasiconcave in $\alpha$, we can apply
efficient convex optimization techniques to obtain the optimal
$\alpha$, e.g., bisection search combined with the interior-point
method \cite{boyd}. In the low and high SNR regimes, we have the following results.
\begin{Proposition}
When $\gamma_0\rightarrow 0$, we have $1+\gamma_{Si}\approx 1$,
$1+\mathcal{B}\approx 1$; and from (\ref{snraf}), the received SNR
is asymptotically given by
\begin{align*}
\gamma_{\text{AF}}\approx \mathbb{E}(\gamma_{Si}\gamma_{iD}) =
\alpha (1-\alpha) \mathbb{E}(\rho_{Si}\rho_{iD}),
\end{align*}
which implies $\alpha_{\text{opt}}=0.5$.
\end{Proposition}

\begin{Remark}
In the low SNR regime, compared with the noise power at the
destination, the amplified noise from the relays can be neglected,
i.e., $\mathcal{B} \approx 0$, and the received SNR can be
approximated as the product of the two SNRs over the two hops.
Obviously, the optimal point is achieved when the powers allocated
to the two hops are the same.
\end{Remark}

\begin{Proposition} \label{nonco_af_high_optpower}
When $\gamma_0 \rightarrow \infty$, we have
$1+\gamma_{Si}\sim \gamma_{Si}$; and from (\ref{snraf}), the
received SNR is asymptotically given by
\begin{align} \label{snrafhigh}
\gamma_{\text{AF}} \approx \frac{(1-p)(1-\alpha)\gamma_0
\mathbb{E}(\rho_{iD})}{1+(1-p)\frac{1-\alpha}{\alpha}\mathbb{E}
\left(\frac{\rho_{iD}}{\rho_{Si}}\right)}.
\end{align}
By letting the derivative of (\ref{snrafhigh}) be zero, it is
observed that only one solution satisfies the condition $\alpha
\in [0,1)$, which is given as
\begin{align} \label{opthigh}
\alpha_{\text{opt}}= \frac{\sqrt{(1-p)\mathbb{E}\left(
\frac{\rho_{iD}}{\rho_{Si}}
\right)}}{1+\sqrt{(1-p)\mathbb{E}\left(
\frac{\rho_{iD}}{\rho_{Si}} \right)}}.
\end{align}
Taking the second-order derivative of (\ref{snrafhigh}), we have
\begin{align*}
\frac{d^2\gamma_{\text{AF}}}{d\alpha^2}= \frac{ -(1-p)^2
\mathbb{E} \left( \frac{\rho_{iD}}{\rho_{Si}} \right)
\gamma_0}{\left( (1-p)\mathbb{E}\left(
\frac{\rho_{iD}}{\rho_{Si}} \right) +\left(1-(1-p)\mathbb{E}\left(
\frac{\rho_{iD}}{\rho_{Si}} \right)\right)\alpha  \right)^2}.
\end{align*}
Since $\frac{d^2\gamma_{\text{AF}}}{d\alpha^2}\leq 0$ such that
$\gamma_{\text{AF}}$ in (\ref{snrafhigh}) is concave, we conclude
that the power allocation factor given in (\ref{opthigh}) is
indeed the optimal solution to maximize $\gamma_{\text{AF}}$.
\end{Proposition}

\begin{Remark}
From (\ref{opthigh}), we see that the optimal power allocation
factor in the high SNR regime is determined by the attack
probability $p$ and the network topology parameter
$\mathbb{E}\left( \frac{\rho_{iD}}{\rho_{Si}}\right)$. Moreover,
for a given relay network, the optimal $\alpha$ is a decreasing
function of $p$ when SNR is high, i.e., when attacks are more
likely, \emph{more power should be allocated to the relays}.
Intuitively, for larger $p$ values, more relays are prone to die, and the
received SNR in the second hop is reduced. Therefore, in order to
balance these two hops, we should allocate more power to the
second hop.
\end{Remark}

\subsubsection{Performance Evaluation}

For a general $\gamma_0$ value, the AF strategy cannot achieve the MAC
cut-set upper bound. However, in the high and low SNR regimes, we
have the following results.
\begin{Proposition} \label{afachie}
When $\gamma_0$ goes to infinity, the $\epsilon$-outage rate with the AF
strategy and the MAC cut-set bound have the following asymptotic
relationship:
\begin{align}
R_{\text{AF}} \sim C_{\text{upper}} - \mathcal{O}(1).
\end{align}
\end{Proposition}

\begin{proof}
When $\gamma_0 \rightarrow \infty$, we have $1+\gamma_{Si} \sim
\gamma_{Si}$. Hence, we have $\mathcal{A} \sim
\mathbb{E}(\gamma_{iD}),~\text{and}~\mathcal{B} \sim \mathbb{E}
\left( \frac{\gamma_{iD}}{\gamma_{Si}} \right)$. Thus, the
achievable rate with the AF strategy can be approximated as
\begin{align*}
R_{\text{AF}} \sim & \frac{1}{2} \log \left(
\frac{(1-p)\mathcal{A}}{1+ (1-p)\mathcal{B}} \ln\left(
\frac{1}{1-\epsilon} \right)\right) \\
\sim &\frac{1}{2} \log\left(
\frac{(1-p)\mathbb{E}(\gamma_{iD})}{1+ (1-p) \mathbb{E} \left(
\frac{\gamma_{iD}}{\gamma_{Si}} \right)}\ln\left(
\frac{1}{1-\epsilon} \right) \right) \\
=&\frac{1}{2}\log \left(
(1-p)\gamma_0\mathbb{E}(\rho_{iD})\ln\left( \frac{1}{1-\epsilon}
\right) \right)  -\frac{1}{2}\log \left(
\frac{1}{1-\alpha}+\frac{1-p}{\alpha}
\mathbb{E} \left( \frac{\rho_{iD}}{\rho_{Si}} \right)\right) \\
=&C_{\text{upper}} - \mathcal{O}(1),
\end{align*}
where $\alpha$ is chosen to minimize the achievable rate loss as
in (\ref{opthigh}). Therefore, the proposition follows.
\end{proof}

\begin{Remark}
The gap $\frac{1}{2}\log \left( \frac{1}
{1-\alpha}+\frac{1-p}{\alpha} \mathbb{E} \left(
\frac{\rho_{iD}}{\rho_{Si}} \right)\right)$ is independent of $\gamma_0$,
and is determined by the power allocation factor $\alpha$, the
attack probability $p$, and the network topology. Here,
$\mathbb{E} \left( \frac{\rho_{iD}}{\rho_{Si}}\right)$ is a
parameter determined by the topology of the network.
\end{Remark}

\begin{Proposition}
In the low SNR regime, the $\epsilon$-outage achievable rate of
the AF strategy demonstrates the following asymptotic property
\begin{align}
\mathop {\lim }\limits_{\scriptstyle {\gamma_0} \to 0 \atop
\scriptstyle \epsilon  \to 0} \frac{R_{\text{AF}}}{\epsilon
\gamma_0} = 0.
\end{align}
\end{Proposition}
\begin{proof}
For the outage probability, we have
\begin{align*}
\mathop {\lim }\limits_{\scriptstyle {\gamma_0} \to 0 \atop
\scriptstyle \varepsilon  \to 0}
\frac{{\epsilon}}{\frac{{R}}{{{\gamma_0}}}} &=  \mathop {\lim
}\limits_{\scriptstyle {\gamma_0} \to 0 \atop \scriptstyle
\varepsilon  \to 0} \frac{\Pr \left\{ \frac{1}{2} \log \left( 1+
\gamma_{\text{AF}}
|h|^2 \right) < R \right\}} {\frac{R}{\gamma_0}} \\
&= \mathop {\lim }\limits_{\scriptstyle {\gamma_0} \to 0 \atop
\scriptstyle \varepsilon  \to 0} \frac{ \frac{2\ln 2
R}{\gamma_{\text{AF}}} }{\frac{R}{\gamma_0}} = \mathop {\lim
}\limits_{\scriptstyle {\gamma_0} \to 0 \atop \scriptstyle
\varepsilon  \to 0}\frac{ \frac{2\ln2R}{(1-p) \mathbb{E} \left(
\gamma_{Si}\gamma_{iD}\right) } } { \frac{R}{\gamma_0}} \\
&= \mathop {\lim }\limits_{\scriptstyle {\gamma_0} \to 0 \atop
\scriptstyle \varepsilon  \to 0}\frac{ \frac{2\ln2R}{(1-p) \alpha
(1-\alpha)\gamma_0^2 \mathbb{E} \left( \rho_{Si}\rho_{iD}\right)
} } {\frac{R}{\gamma_0}} =\infty.
\end{align*}
Therefore, the proposition follows.
\end{proof}

\begin{Remark}
This proposition shows that the $\epsilon$-outage rate of the AF
strategy decreases at a higher speed than that of $\gamma_0$ in the low SNR
regime. The reason why the AF strategy performs so bad in the low
SNR regime is that in this case the relays mainly forward noises,
but not the signal, which is true regardless of the fact that we have physical node failures or not.
\end{Remark}

\subsection{DF Strategy} \label{non_df}

In this subsection, we focus on the $\epsilon$-outage rate of the
DF strategy. With Rayleigh fading in the source-to-relay links, the
probability that the relay cannot decode the source message
successfully is strictly positive. As such, in addition to the
probability of physical random attacks, we should also take into
account the fading-related decoding failures at the relay.

\subsubsection{Achievable Rate}
With the DF strategy and equal transmission power assumptions
across the relays, the received signal at the destination is given
by
\begin{align} \label{dfreceveived}
y=\underbrace{\sum_{i=1}^{L} \sqrt{\frac{(1-\alpha) P
\rho_{iD}}{N} } h_{iD}}_{M} x + w,
\end{align}
where $L$ is a binomial random variable with parameters $(N,p_0
(1-p))$, and $p_0$ is the average probability that the relay can
successfully decode the source message. Specifically, $p_0$ is
computed as
\begin{align} \label{p0}
p_0 = \mathbb{E}\left[ p_0(\mathbf{s}) \right]=
\int_{\mathbf{s}\in \mathcal{S}} p_0(\mathbf{s}) p
(\mathbf{s})d\mathbf{s},
\end{align}
where
\begin{align} \label{relaydying}
p_0(\mathbf{s}) =\Pr \bigg\{ \frac{1}{2} \log \left( 1+\alpha
\gamma_0 \rho_{Si}\left( \mathbf{s} \right) |h_{Si}|^2 \right)
\geq R \bigg\} = \exp \left( - \frac{2^{2R}-1}{\alpha \gamma_0
\rho_{Si}\left( \mathbf{s} \right)} \right),
\end{align}
with $R$ being the target transmission rate.

To compute the outage probability of this network, we need to know
the distribution of $M$ defined in (\ref{dfreceveived}). For a different
path-loss $\rho_{iD}\left( \mathbf{s} \right)$, the probability of
decoding failure defined in (\ref{relaydying}) is different, which
is jointly determined by the location of the relay and the fading
degree. To calculate the PDF of successful relay decoding
over different locations of $\mathbf{s}$, we have the following
proposition.

\begin{Proposition} \label{non_df_surving}
Without considering the node failures caused by random attacks,
the PDF of successful relay decoding at a given location $\mathbf{s}$
is given as
\begin{align} \label{pdf_surviving}
f(\mathbf{s}) = \frac{1}{  p_0 }p(\mathbf{s})p_0(\mathbf{s}),~
\mathbf{s} \in \mathcal{S}.
\end{align}
\end{Proposition}
\begin{proof}
See Appendix \ref{app_non_df_surviving}.
\end{proof}

Similar to the analysis in the previous subsection, when $N$ goes
to infinity, we know that $M$ is asymptotically Gaussian
\cite{robbins} with $M\sim \mathcal{CN}\left(0,p_0(1-p)(1-\alpha)
\gamma_0 \mathbb{E}_1 (\rho_{iD}) \right)$, where $\mathbb{E}_1
(\rho_{iD})$ is defined as
\begin{align} \label{exp1}
\mathbb{E}_1(\rho_{iD}) = \int_{\mathbf{s}\in \mathcal{S}}
\rho_{iD}(\mathbf{s}) f(\mathbf{s}) d \mathbf{s}.
\end{align}
From the analysis above, we observe that for the DF strategy, the
relay network under consideration is also equivalent to a Rayleigh
fading channel. Therefore, the outage probability with the DF
strategy is
\begin{align} \label{pdf}
p_{\text{DF}} = 1-\exp \left( -\frac{2^{2R}-1}{\gamma_{\text{DF}}}
\right),
\end{align}
where $\gamma_{\text{DF}}$ is the average received SNR at the
destination, and is given as
\begin{align} \label{snrdf}
\gamma_{\text{DF}} = (1-p)(1-\alpha) \gamma_0
\int_{\mathbf{s}\in \mathcal{S}} \rho_{iD}
p(\mathbf{s})p_0(\mathbf{s}) d \mathbf{s}.
\end{align}

Now, we could define the $\epsilon$-outage rate of the DF strategy
as
\begin{align} \label{erdf}
R_{\text{DF}}= \max_{R} \left\{ R: p_{\text{DF}} \leq \epsilon
\right\}.
\end{align}
Generally, it is quite challenging to derive a closed-form for
$R_{\text{DF}}$ in terms of $\epsilon$ and SNR. However, in
practical communication systems, $\epsilon$ is usually set at very
small values, which implies that $\gamma_0 \gg 2^{2R}-1$ is
required. With this assumption, we can approximate
(\ref{relaydying}) and (\ref{pdf}) as
\begin{align}
&p_0(\mathbf{s}) \approx 1- \frac{2^{ 2 R }-1}{\alpha
\gamma_0\rho_{Si}}, \label{ap0}
\\
&p_{\text{DF}} \approx \frac{2^{ 2 R }-1}{\gamma_{\text{DF}}}
\label{apdf},
\end{align}
respectively. Substituting (\ref{snrdf}) and (\ref{ap0}) to (\ref{apdf}), we calculate
the $\epsilon$-outage rate defined in (\ref{erdf}) as
\begin{align} \label{ardf}
R_{\text{DF}} = \frac{1}{2}\log \left( 1+
\frac{(1-p)(1-\alpha)\gamma_0
\mathbb{E}(\rho_{iD})\epsilon}{1+\epsilon
(1-p)\frac{1-\alpha}{\alpha} \mathbb{E} \left(
\frac{\rho_{iD}}{\rho_{Si}} \right)} \right).
\end{align}

\begin{Remark} \label{df_noise}
Based on (\ref{snrdf}), (\ref{ap0}), and (\ref{apdf}), and after some mathematical manipulations, it is shown that the term $\epsilon
(1-p)\frac{1-\alpha}{\alpha} \mathbb{E} \left(
\frac{\rho_{iD}}{\rho_{Si}} \right)$ in (\ref{ardf}) is
mainly due to the fading-related decoding failures (the second term of the right-hand side of (\ref{ap0})); and if $p_0(\mathbf{s})=1$ standing for decoding failures caused by fading, the above term in (\ref{ardf}) degrades to zero. This suggests that losing part of the
relay nodes due to decoding failures is equivalent to increasing a
certain amount of noise at the destination.
\end{Remark}

\begin{Remark}
From (\ref{capacityupper}), Proposition \ref{afachie}, and (\ref{ardf}), we observe that as $\gamma_0$ gets large, the MAC cut-set bound, the AF rate, and the DF rate all scale as $\mathcal{O}\left( \log \left( \gamma_0 \right) \right)$.
\end{Remark}

With the small outage probability assumption adopted above, some
approximation errors may be incurred in the achievable rate
expression. In Fig. \ref{df_appro}, we show the comparison between
the exact DF achievable rate and the approximate rate based on
(\ref{ap0}) and (\ref{apdf}), where the exact rate is numerically
computed by using (\ref{pdf}) and (\ref{snrdf}). All the
parameters used to draw these figures are the same as those later
defined in Section \ref{non_simu}. It is observed that in both the
high and low SNR regimes, the approximation works well; and
as $\epsilon$ gets smaller, the approximation gap gets smaller.

\subsubsection{Effect of Random Attacks}

Based on (\ref{ardf}), the DF rate losses due to random attacks in
both the low and high SNR regimes are evaluated next.

\begin{Proposition}
In the low SNR regime, there is a less than $p$-portion rate loss
for the DF strategy.
\end{Proposition}
\begin{proof}
We have
\begin{align}
R^{\text{Low}} & \approx \frac{1}{2} \cdot
\frac{1-p}{1+(1-p)\epsilon \frac{1-\alpha}{\alpha} \mathbb{E}
\left( \frac{\rho_{iD}}{\rho_{Si}}\right) } (1-\alpha)\gamma_0
\mathbb{E}(\rho_{iD})\epsilon \\
& = (1-p) \cdot \frac{1+\epsilon \frac{1-\alpha}{\alpha}
\mathbb{E} \left(
\frac{\rho_{iD}}{\rho_{Si}}\right)}{1+(1-p)\epsilon
\frac{1-\alpha}{\alpha} \mathbb{E} \left(
\frac{\rho_{iD}}{\rho_{Si}}\right)} \cdot R_{\text{0}}^{\text{Low}} \label{non_df_loss_low} \\
&\geq (1-p) R_{\text{0}}^{\text{Low}}
\end{align}
where $R_{\text{0}}^{\text{Low}} = \frac{1}{2} \cdot
\frac{(1-\alpha)\gamma_0 \mathbb{E}(\rho_{iD})\epsilon}
{1+\epsilon \frac{1-\alpha}{\alpha} \mathbb{E} \left(
\frac{\rho_{iD}}{\rho_{Si}}\right) } $ is the achievable rate of
the DF strategy in the low SNR regime with $p=0$. Since the second
term in (\ref{non_df_loss_low}) is larger than one, the achievable
rate loss due to random attacks is smaller than $p$-portion
compared to the attack-free case. Only when $\epsilon$ goes to
zero, or $p$ goes to zero, the percentage of the rate loss is
about $p$-portion.
\end{proof}

\begin{Remark}
The percentage of achievable rate loss for the DF strategy is less
than that of the AF strategy. The reason is that due to random
attacks, less portion of relays suffer from decoding failures,
i.e., $(1-p)$-portion compared to the attack-free case. Then, from
the argument of Remark \ref{df_noise}, it is known that this is effectively
equivalent to adding less noises to the receiver. Therefore, the
achievable rate loss is less than $p$-portion, even we lose $p$-portion of relay nodes.
\end{Remark}

\begin{Proposition}
In the high SNR regime, there is a constant rate loss for the DF
strategy, and the constant is upper-bounded by $\frac{1}{2}\log
\left( \frac {1}{1-p} \right)$.
\end{Proposition}
\begin{proof}
We have
\begin{align}
R^{\text{High}} & \approx \frac{1}{2}\log \left(
\frac{(1-\alpha)\gamma_0
\mathbb{E}(\rho_{iD})\epsilon}{1+\epsilon \frac{1-\alpha}{\alpha}
\mathbb{E} \left( \frac{\rho_{iD}}{\rho_{Si}} \right)} \right)
-\frac{1}{2}\log \left( \frac{1}{1-p} \right) +
\underbrace{\frac{1}{2} \log \left( \frac{ \left(1+\epsilon
\frac{1-\alpha}{\alpha} \mathbb{E} \left(
\frac{\rho_{iD}}{\rho_{Si}} \right) \right)}{1+\epsilon (1-p)
\frac{1-\alpha}{\alpha} \mathbb{E} \left(
\frac{\rho_{iD}} {\rho_{Si}} \right)} \right)}_{C_2} \label{non_df_loss_high}\\
& \geq R_{\text{0}}^{\text{High}} - \frac{1}{2}\log \left( \frac
{1}{1-p} \right),
\end{align}
where $R_{\text{0}}^{\text{High}}= \frac{1}{2}\log \left(
\frac{(1-\alpha)\gamma_0
\mathbb{E}(\rho_{iD})\epsilon}{1+\epsilon \frac{1-\alpha}{\alpha}
\mathbb{E} \left( \frac{\rho_{iD}}{\rho_{Si}} \right)} \right) $.
Therefore, the achievable rate loss is $R_{\text{0}}^{\text{High}}
- R^{\text{High}}=\frac{1}{2}\log \left( \frac {1}{1-p}
\right)-C_2 \leq \frac{1}{2}\log \left( \frac {1}{1-p} \right)$,
and the equality holds only when $\epsilon$ or $p$ goes to zero.
\end{proof}

\begin{Remark}
Compared with the constant term $C_1$ in (\ref{non_af_high_1}),
the term $C_2$ in (\ref{non_df_loss_high}) is further determined
by $\epsilon$. Since $C_2$ goes to zero as $\epsilon$ tends to
zero and $\epsilon$ is usually much smaller than 1, the absolute
rate loss for the DF strategy is larger than that of the AF strategy given in (\ref{non_af_high_1}).
Moreover, in the high SNR regime, there is a constant gap
between the AF and DF rates, which can be neglected as $\gamma_0$ increases. Therefore,
when using $1-\frac{R}{R_0}$ as the metric, we conclude that in
the high SNR regime, the AF strategy is less sensitive to random
attacks than the DF strategy.
\end{Remark}

\subsubsection{Power Allocation to the Source} \label{power_noncoherent}

In this subsection, we derive the optimal power allocation scheme
for the DF strategy by maximizing the achievable rate in
(\ref{ardf}) over $\alpha$.
\begin{Proposition}
The achievable rate defined in (\ref{ardf}) is a concave function
over $\alpha$.
\end{Proposition}
\begin{proof}
The proof is similar to Proposition \ref{nonco_af_high_optpower},
thus skipped.
\end{proof}

Taking the first-order derivative of (\ref{ardf}), we find that
there is only one solution over $\alpha \in [0,1)$, and the
corresponding optimal power allocation factor is given as
\begin{align} \label{optpowerDF}
\alpha_{\text{opt}} =\frac{\sqrt{\epsilon (1-p)\mathbb{E}\left(
\frac{\rho_{iD}}{\rho_{Si}} \right)}}{1+\sqrt{\epsilon
(1-p)\mathbb{E}\left( \frac{\rho_{iD}}{\rho_{Si}} \right)}}.
\end{align}

\begin{Remark}
As we see from (\ref{optpowerDF}), it is easy to observe that when either $p$
increases, or $\epsilon$ decreases, we should \emph{allocate more
power to the relay nodes}, which is similar to the case with the
AF strategy.
\end{Remark}

\subsubsection{Asymptotic Performance}

We now compare the DF rate against the capacity upper bound at
asymptotically small outage probabilities. To obtain such
asymptotic results, we study the behavior of a normalized rate
$\frac{2^{2C}-1}{\epsilon \gamma_0}$ when $\epsilon$ goes to
zero.

\begin{Proposition} \label{non_df_optimal}
As $\epsilon$ goes to zero, the DF strategy can achieve the
capacity upper bound asymptotically, and the asymptotic behavior
of the capacity is given as
\begin{align} \label{lowup}
\mathop {\lim }\limits_{\epsilon \to 0} \frac{2^{2C}-1}{\epsilon
\gamma_0} = (1-p)(1-\alpha)\mathbb{E}\left( \rho_{iD} \right).
\end{align}
\end{Proposition}

\begin{proof}
See Appendix \ref{app_non_df_asymp}.
\end{proof}

\begin{Remark}
This theorem shows that the DF strategy is asymptotically optimal
when $\epsilon$ goes to zero. With the small outage probability
assumption, the $\epsilon$-outage capacity of the so-defined large
fading relay network is approximately given as
\begin{align} \label{noncoherentdflowsnr}
C \approx \frac{1}{2} \log \left( 1+ (1-p) (1-\alpha)\gamma_0
\mathbb{E} (\rho_{iD}) \epsilon \right) .
\end{align}
\end{Remark}

\section{Simulation and Numerical Results} \label{non_simu}

In this section, we present several simulation and numerical
examples to validate our analysis. The following setup is
deployed: The source, destination, and relays are on a straight line, i.e., we consider the 1-D networking case; the locations of the source and destination are $s_S=0$ and
    $s_D=12$, respectively; relays are uniformly located over a line segment
    $[1,11]$; the path loss exponent $\theta=2$. Note that under these conditions, the dead zone assumption is
satisfied, and the following figures are drawn according to the
analytical results derived in previous sections.

In Fig. \ref{non_appro}, we compare the simulation based and
Gaussian approximation based outage probabilities for both the AF
and DF strategies. We observe that when the number of relays is
large, Gaussian approximation works very well.

In Fig. \ref{non_rate}, we draw the asymptotic $\epsilon$-outage
rates of the AF and DF strategies and the MAC cut-set bound for the
large fading relay network with the optimal power allocation between
the source node and all the relays: For the MAC upper bound,
$\alpha=0$; for the AF and DF strategies, the optimal $\alpha$ is
computed based on Proposition \ref{opt_alpha_non_af} or
(\ref{optpowerDF}), respectively. From this figure, it is observed
that as $\gamma_0$ increases, the gaps between the achievable rates and the
upper bound turn to be constants. In particular, for the case
$\epsilon=0.1$, the gap between the upper bound and the DF
achievable rate is approximately 0.9 bit; whereas, the gap between
the upper bound and the AF rate is about 1.9 bit.

In Fig. \ref{non_rate_loss}, we plot the rate loss caused by the
random attacks for both the AF and DF strategies. In the low SNR
regime, the DF strategy suffers a smaller relative rate loss than the AF
strategy: a less than
$p$-portion loss for the DF strategy vs. a $p$-portion loss for the AF strategy. In the high SNR regime, it
is observed that AF suffers a less percentage rate loss than the DF
strategy.

\section{Conclusion}

In this paper, we studied the asymptotic
capacity upper bound and achievable rates with the AF and DF
strategies for a fading relay
networks under random attacks, when the number of relays is infinitely large. We
considered the non-coherent relaying scheme, which corresponds to the case without forward-link CSI at the relays. We proved that the DF strategy is
asymptotic optimal when the outage probability goes to zero, while
the AF strategy is strictly suboptimal in all SNR regimes. We also derived the optimal power allocation factor between the source and the relays in some special scenarios. Regarding the rate loss due to random attacks, we showed that the AF
strategy is relatively less sensitive to random attacks than the DF strategy
in the high SNR regime, while DF performs in a more robust way in the low SNR regime.

\appendices

\section{Proof of Proposition \ref{prop_non_upp}} \label{app_non_upp}
Consider the MAC cut-set in the relay network, and assume that the
surviving (for the physical attacks) relays perfectly decode the source message and transmit
$t_i=\sqrt{\frac{(1-\alpha)P}{N}}x,~0 \leq i \leq L$. The received
signal at the destination is given as
\begin{eqnarray} \label{nonco_upper}
y=\underbrace{\sum_{i=1}^{L} \sqrt{\frac{(1-\alpha) P
\rho_{iD}}{N} } h_{iD}}_{A} x + w.
\end{eqnarray}
From \cite{robbins}, when $N \rightarrow \infty$, we know that $A$ is an
asymptotically complex Gaussian random variable with $A \sim
\mathcal{CN} \left( 0,(1-p) (1-\alpha)P \mathbb{E}(\rho_{iD})
\right)$. As such, the overall source-relays-destination
transmission is over an equivalent Rayleigh fading channel.
Correspondingly, the lower bound of the outage probability is
\begin{eqnarray}
\epsilon = \Pr \bigg\{ \frac{1}{2} \log \left( 1+
\gamma_{\text{upper}} |h|^2 \right)< C \bigg\}, \nonumber
\end{eqnarray}
where $C$ is the target rate, $h$ is a standard complex Gaussian
random variable, and the average received SNR
$\gamma_{\text{upper}}$ is defined as in (\ref{snrupper}). Then,
based on the outage capacity definition given in Chapter 5 of
\cite{tse}, $C
:=\log\left(1+\mathcal{F}^{-1}(1-\epsilon)\gamma_0 \right)$,
where $\mathcal{F}^{-1}(x)$ is the inverse function of
$\mathcal{F}(x):=\Pr\{ |h|^2 \leq x \}$. By computing
$\mathcal{F}(x)$, and subsequently $\mathcal{F}^{-1}(x)$, we
obtain the outage capacity upper bound given in
(\ref{capacityupper}).

\section{Proof of Proposition \ref{opt_alpha_non_af}} \label{app_non_af_power}
First, we have the following results for $\mathcal{A}$ and
$\mathcal{B}$.
\begin{enumerate}
    \item \label{aaa} For $\mathcal{A}$, we consider the following function
\begin{align*}
f(\alpha,\mathbf{s})& =
\frac{\alpha(1-\alpha)\gamma_0^2\rho_{Si}\rho_{iD}}{1+\alpha
\gamma_0\rho_{Si}}
 = \gamma_0^2\rho_{Si}\rho_{iD}\left(c_1\alpha+c_2+\frac{c_3} {1+\alpha
\gamma_0\rho_{Si}}\right),
\end{align*}
where $c_1,~c_2$, and $c_3$ are some constants, and
$c_3=-\left(\frac{1}{\gamma_0\rho_{Si}}+\frac{1}{\gamma_0^2\rho_{Si}^2}\right)<0$.
Since $c_1\alpha+c_2$ is affine and
$\frac{c_3}{1+\alpha\gamma_0\rho_{Si}}$
 is concave, $f(\alpha,\mathbf{s})$ is concave in $\alpha$ for any given $\mathbf{s}$.
 Moreover, (\ref{A}) is the integral over $\mathbf{s}$, which does not
 change the concavity of the original function. Therefore,
 (\ref{A}) is also concave in $\alpha$.

    \item \label{bbb}For $\mathcal{B}$, we consider the following
    function
\begin{align*}
g(\alpha,\mathbf{s})=\frac{(1-\alpha)\gamma_0 \rho_{iD} }{1+
\alpha\gamma_0 \rho_{Si}} \nonumber
=\frac{\rho_{iD}}{\rho_{Si}}\left( \frac{1}{1+\alpha \gamma_0
\rho_{Si}} -1 \right).
\end{align*}
Since $g(\alpha,\mathbf{s})$ is convex in $\alpha$, by a similar
argument as in \ref{aaa}), we have that $1+\mathcal{B}$ is convex
in $\alpha$.
\end{enumerate}

From \ref{aaa}) and \ref{bbb}), we know that $\mathcal{A}>0$ is
concave and $1+\mathcal{B}>0$ is convex. Based on Example 3.38 in
\cite{boyd}, we know that $\frac{1+\mathcal{B}}{\mathcal{A}}$ is a
quasiconvex function. Therefore, we conclude that
$\gamma_{AF}=\frac{\mathcal{A}}{1+\mathcal{B}}$ is a quasiconcave
function over $\alpha$.

\section{Proof of Proposition \ref{non_df_surving}} \label{app_non_df_surviving}
For any subset $\mathcal{D} \subseteq \mathcal{S}$, the
probability that the $i$-th relay node is located in $\mathcal{D}$
and can successfully decode the source message, can be computed as
\begin{align}
\Pr\{\text{relay $i$ } \in ~\mathcal{D},~\text{and relay $i$
decodes successfully}\} = \int_{\mathbf{t} \in \mathcal{D}}
p(\mathbf{t}) p_0(\mathbf{t}) d\mathbf{t}.
\end{align}

Therefore, we have the conditional probability function
\begin{align}
&~~~~\Pr\{\text{relay $i$ }\in ~\mathcal{D} \big| \text{relay $i$
decodes successfully} \} \nonumber \\
& =  \frac{\Pr\{\text{relay $i$ } \in ~\mathcal{D},~\text{and
relay $i$ decodes successfully}\} } {\Pr \{ \text{relay $i$
decodes successfully} \}} \nonumber \\
& = \frac{\int_{\mathbf{t} \in \mathcal{D}} p(\mathbf{t})
p_0(\mathbf{t}) d\mathbf{t}}{\int_{\mathbf{s}\in \mathcal{S}}
p(\mathbf{s})p_0(\mathbf{s})d \mathbf{s}}  =\frac{1}{p_0}\int_{\mathbf{t} \in \mathcal{D}} p(\mathbf{t})
p_0(\mathbf{t}) d\mathbf{t}. \label{pcf_surviving}
\end{align}
which is the probability of the successful decoding relay node
being located in area $\mathcal{D}$. Taking derivative of
(\ref{pcf_surviving}), we have the corresponding PDF given as
(\ref{pdf_surviving}).

\section{Proof of Proposition \ref{non_df_optimal}} \label{app_non_df_asymp}
First we derive the asymptotic behavior of the outage probability
$\epsilon$ for the MAC cut-set upper bound, and we have
\begin{align*}
\mathop {\lim }\limits_{ \epsilon  \to 0} \frac{{ \epsilon }}
{\frac{2^{2C}-1}{\gamma_0}} &=\mathop {\lim }\limits_{ \epsilon
\to 0} \frac{ \Pr \left\{ \frac{1}{2} \log\left( 1+
\gamma_{\text{upper}} |h_{iD}|^2 \right) < C \right\}
} {\frac{2^{2C}-1}{\gamma_0}} \\
&= \mathop {\lim }\limits_{\frac{2^{2C}-1}{\gamma_0} \to 0}
\frac{1-\exp \left( - \frac{2^{2C}-1 }{
\gamma_{\text{upper}}}\right) } {\frac{2^{2C}-1}{\gamma_0}} \\
& = \frac{1}{(1-p)(1-\alpha)\mathbb{E}\left( \rho_{iD} \right)}.
\end{align*}

Next, we consider the asymptotic behavior of the outage
probability for the DF strategy, and we have
\begin{align*}
\mathop {\lim }\limits_{ \scriptstyle \varepsilon  \to 0} \frac{{
\epsilon }}{\frac{2^{2C}-1}{\gamma_0}} &=\mathop {\lim
}\limits_{ \scriptstyle \frac{2^{2C}-1}{\gamma_0}  \to 0}
\frac{1}{(1-p)(1-\alpha) \int_{\mathbf{s} \in \mathcal{S} }
\rho_{iD} p(\mathbf{s})p_0
(\mathbf{s})d\mathbf{s}} \\
&= \mathop {\lim }\limits_{ \scriptstyle
\frac{2^{2C}-1}{\gamma_0}  \to 0} \frac{1}{(1-p)(1-\alpha)\left(
\mathbb{E} \left( \rho_{iD} \right) - \frac{2^{2C}-1}{\alpha
\gamma_0 \mathbb{E} \left( \frac{\rho_{iD}}{\rho_{Si}} \right)}
\right) }   \\
&= \frac{1}{(1-p)(1-\alpha) \mathbb{E}\left( \rho_{iD} \right)}.
\end{align*}
The proposition follows immediately.




%

\begin{figure}[h]
\centering \scalebox{.9}{\includegraphics{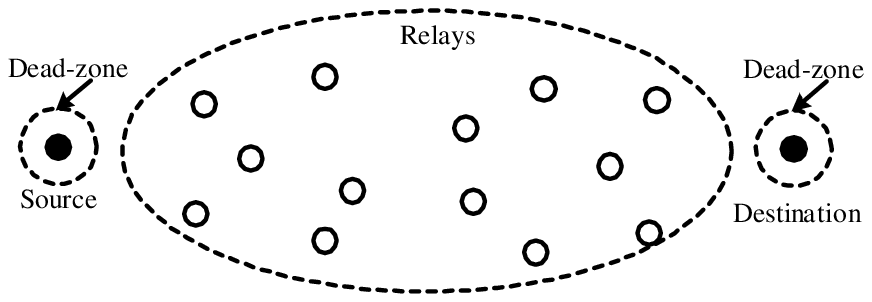}} \caption{Large
fading relay networks with randomly deployed relay nodes and
transmission dead-zone.} \label{fig1}
\end{figure}

\begin{figure}[h]
     \centering
     \subfigure[Low SNR regime]{
          \label{df_appro1}
          \includegraphics[width=.7\linewidth]{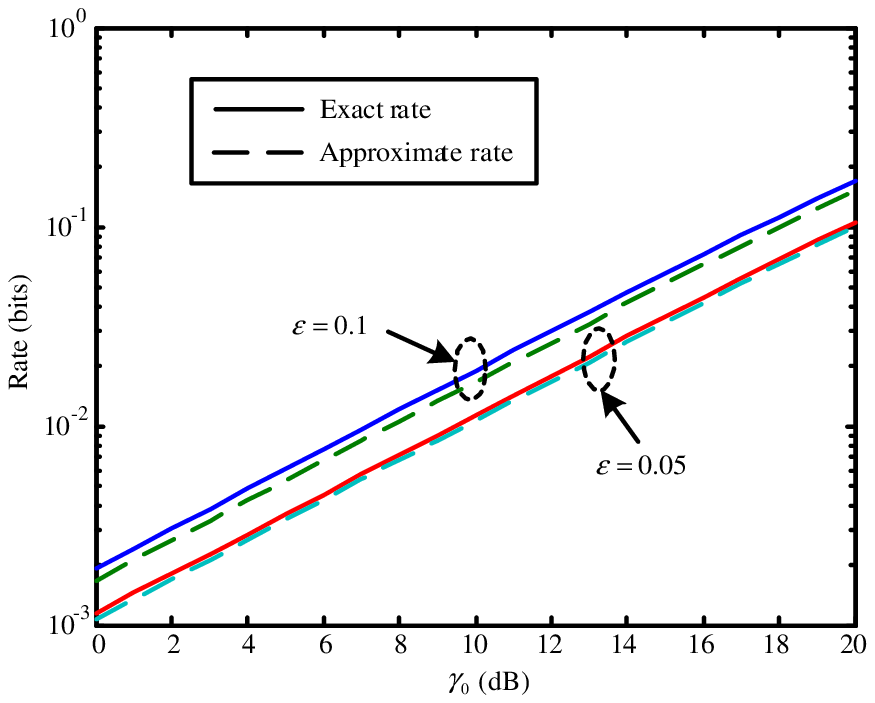}}
     \hspace{.3in}
     \subfigure[High SNR regime]{
          \label{af_appro2}
          \includegraphics[width=.7\linewidth]{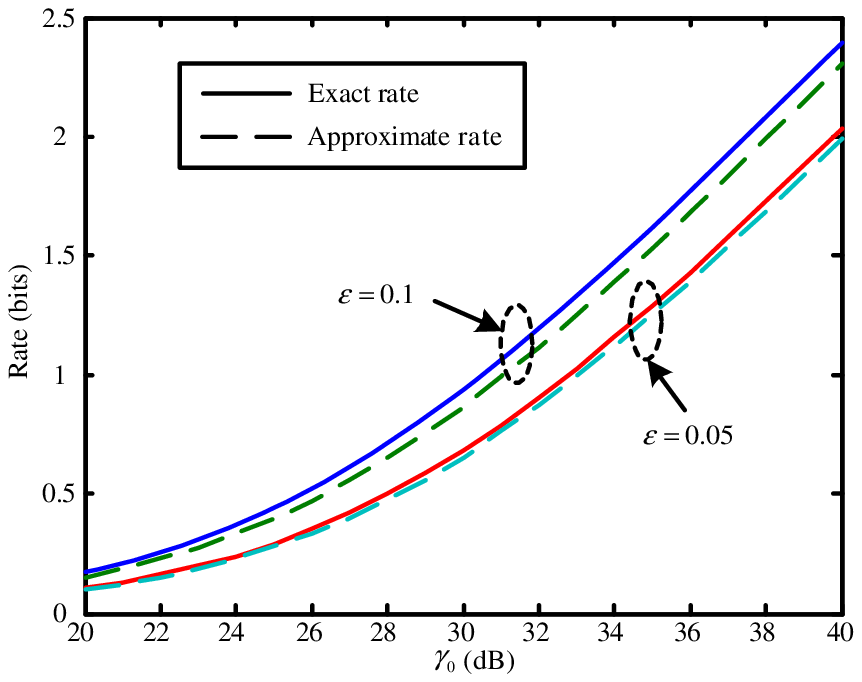}}
     \caption{Comparison of the exact rate and the approximated rate for the DF strategy using the small outage probability approximation}
     \label{df_appro}
\end{figure}

\begin{figure}[h]
\centering
\includegraphics[width=.7\linewidth]{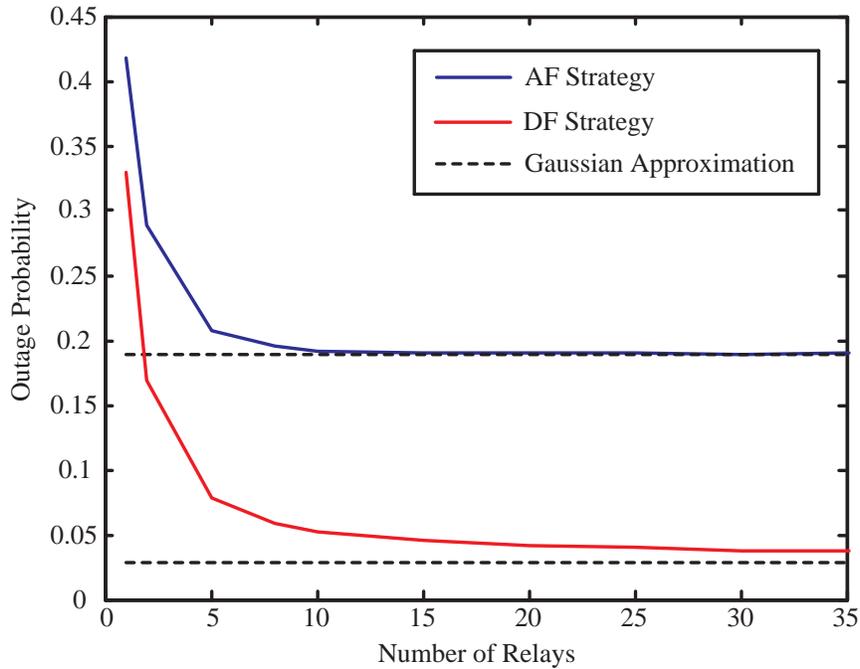}
\caption{Comparison of the outage probabilities with Gaussian
approximation and simulations, $\alpha=0.5$, $p=0.2$, $\gamma_0=30$ dB.}
\label{non_appro}
\end{figure}

\begin{figure}[h]
\centering
\includegraphics[width=.7\linewidth]{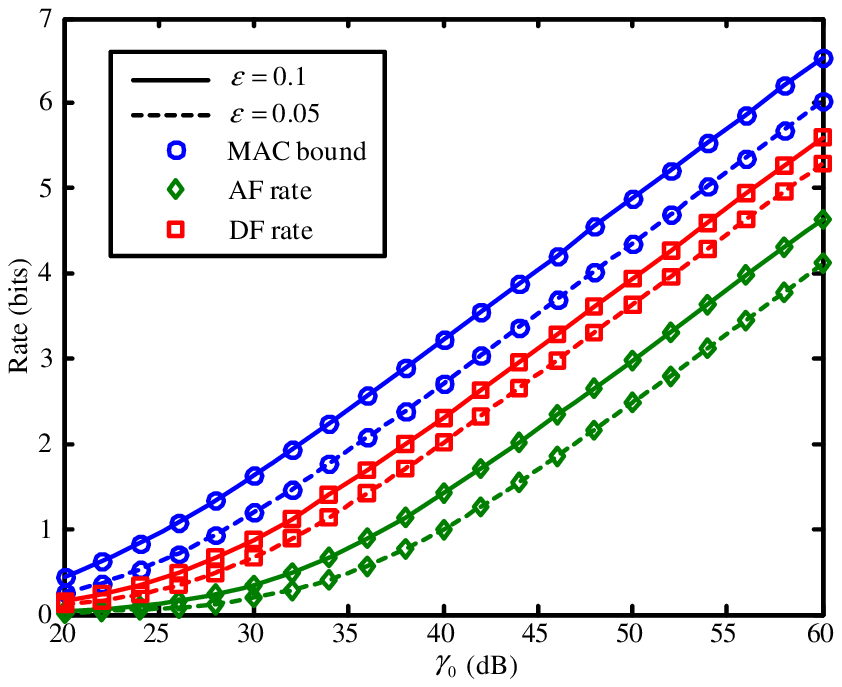}
\caption{Capacity upper bound, the AF rate, and the DF rate for
noncoherent relay networks with optimal power allocation,
$p=0.1$.} \label{non_rate}
\end{figure}

\begin{figure}[h]
\centering
\includegraphics[width=.7\linewidth]{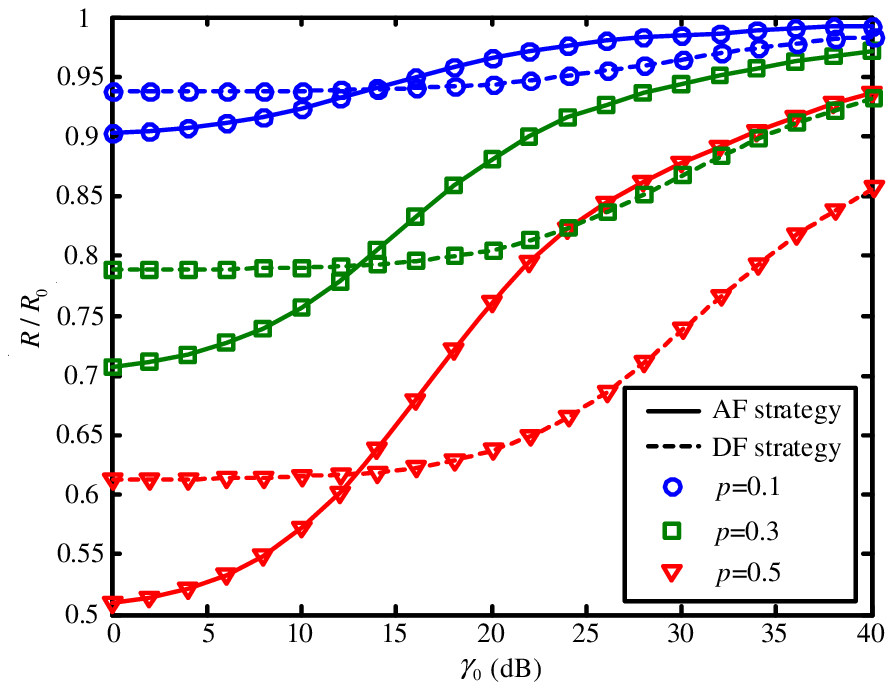}
\caption{Rate loss due to random attacks, $\alpha=0.6$.}
\label{non_rate_loss}
\end{figure}

\end{document}